\documentclass[onecolumn]{IEEEtran}
\IEEEoverridecommandlockouts

\usepackage{setspace}
\usepackage{cite}
\usepackage{amsmath,amssymb,amsfonts,amsthm}
\usepackage{algorithmic}
\usepackage{graphicx}
\usepackage{textcomp}
\usepackage{xcolor}
\usepackage{todonotes}
\usepackage{mathrsfs}
\usepackage[normalem]{ulem}
\usepackage[inline]{enumitem}
\usepackage{xspace}
\usepackage[capitalize]{cleveref}

\def\BibTeX{{\rm B\kern-.05em{\sc i\kern-.025em b}\kern-.08em
    T\kern-.1667em\lower.7ex\hbox{E}\kern-.125emX}}

\definecolor{bleudefrance}{rgb}{0.19, 0.55, 0.91}

\newtheorem{lemma}{Lemma}

\newtheorem{definition}{Definition}
\newtheorem{theorem}{Theorem}
\newtheorem{example}{Example}

\newcommand{\bl}{n}
\newcommand{\numedits}{t}
\newcommand{\editsize}{k}
\newcommand{\confball}[3]{\mathcal{C}_{#1}^{#2}(#3)}
\newcommand{\resconfball}[3]{\tilde{\mathcal{C}}_{#1}^{#2}(#3)}
\newcommand{\editball}[3]{\mathcal{B}_{#1}^{#2}(#3)}
\newcommand{\length}[1]{{\ell}(#1)}
\newcommand{\encode}{\mathsf{Enc}}

\newcommand{\insdel}{IDS\xspace}

\newcommand{\tksub}{$\numedits$ $\editsize$-substring\xspace}
\newcommand{\defeq}{\ensuremath{\triangleq}}

\newcommand{\bfc}{{\mathbf{ c}}}

\newcommand{\bfp}{{\mathbf{ p}}}

\newcommand{\bfu}{{\mathbf{ u}}}
\newcommand{\bfv}{{\mathbf{ v}}}

\newcommand{\bfx}{{\mathbf{ x}}}
\newcommand{\bfy}{{\mathbf{ y}}}

\newcommand{\cC}{\mathcal{C}}

\newcommand{\cP}{\mathcal{P}}

\newcommand{\cS}{\mathcal{S}}

\usepackage{tikz}

\makeatletter
\NewCommandCopy\@@pmod\pmod
\DeclareRobustCommand{\pmod}{\@ifstar\@pmods\@@pmod}
\def\@pmods#1{\mkern4mu({\operator@font mod}\mkern 6mu#1)}
\makeatother

\begin{document}
\title{Coding Schemes for Document Exchange \\ under Multiple Substring Edits
\thanks{This work is supported by the German Research Foundation (DFG) under Grant numbers BI~2492/5-1 and WA 3907/12-1.}
}

\author{%
	\IEEEauthorblockN{
		Hrishi Narayanan, 
		Vinayak Ramkumar, Rawad Bitar, Antonia Wachter-Zeh \\
	}
	\IEEEauthorblockA{
	 Technical  University of Munich, Germany \\
	}
	\IEEEauthorblockA{
		\{hrishi.narayanan, vinayak.ramkumar, rawad.bitar, antonia.wachter-zeh\}@tum.de
	}
}


\maketitle

\begin{abstract}
We study the document exchange problem under multiple substring edits. A substring edit in a string $\bfx$ occurs when a substring $\bfu$ of $\bfx$ is replaced by an arbitrary string $\bfv$.
The lengths of $\bfu$ and $\bfv$ are bounded from above by a fixed constant. Let $\mathbf{x}$ and $\mathbf{y}$ be two binary strings that differ by multiple substring edits. 
The aim of document exchange schemes is to construct an encoding of $\mathbf{x}$ 
with small length such that $\mathbf{x}$ can be recovered using $\mathbf{y}$ and the encoding. 
We construct a low-complexity document exchange scheme with encoding length of $4\numedits\log\bl + o(\log\bl)$ bits, where $n$ is the length of the string $\mathbf{x}$. The best known scheme achieves an encoding length of $4\numedits\log\bl + O(\log\log\bl)$ bits, but at a much higher computational complexity.
Then, we investigate the average length of valid encodings for document exchange schemes with uniform strings $\mathbf{x}$ and develop a scheme with an expected encoding length of
$(4\numedits-1)\log\bl + o(\log\bl)$ bits. In this setting, prior works have only constructed schemes for a single substring edit. 

\end{abstract}


\section{Introduction}

The document exchange problem, also known as data synchronization, involves two nodes $A$ and $B$ storing strings $\mathbf{x}$ and $\mathbf{y}$, respectively, where the string $\mathbf{y}$ 
differs slightly from $\mathbf{x}$. The goal is for node $B$ to recover $\mathbf{x}$ by receiving a small amount of information from $A$ and leveraging the knowledge of $\mathbf{y}$. We consider the scenario in which $\mathbf{x}$ and $\mathbf{y}$ are binary strings and node $A$ computes an encoding of the string $\mathbf{x}$, denoted by $\encode(\mathbf{x})$, and communicates it to $B$. The encoding is designed in such a way that node $B$ is able to recover $\mathbf{x}$ from $\mathbf{y}$ and $\encode(\mathbf{x})$. The length of $\encode(\bfx)$ is termed \emph{redundancy}.
Since the goal is for node $B$ to obtain $\mathbf{x}$ with minimum communication between the nodes, the redundancy should be made as small as possible. See Fig.~\ref{fig:dx} for an illustration of the document exchange setting.

\begin{figure}[htbp]  
	\centering
    \includegraphics[width=0.5\columnwidth]{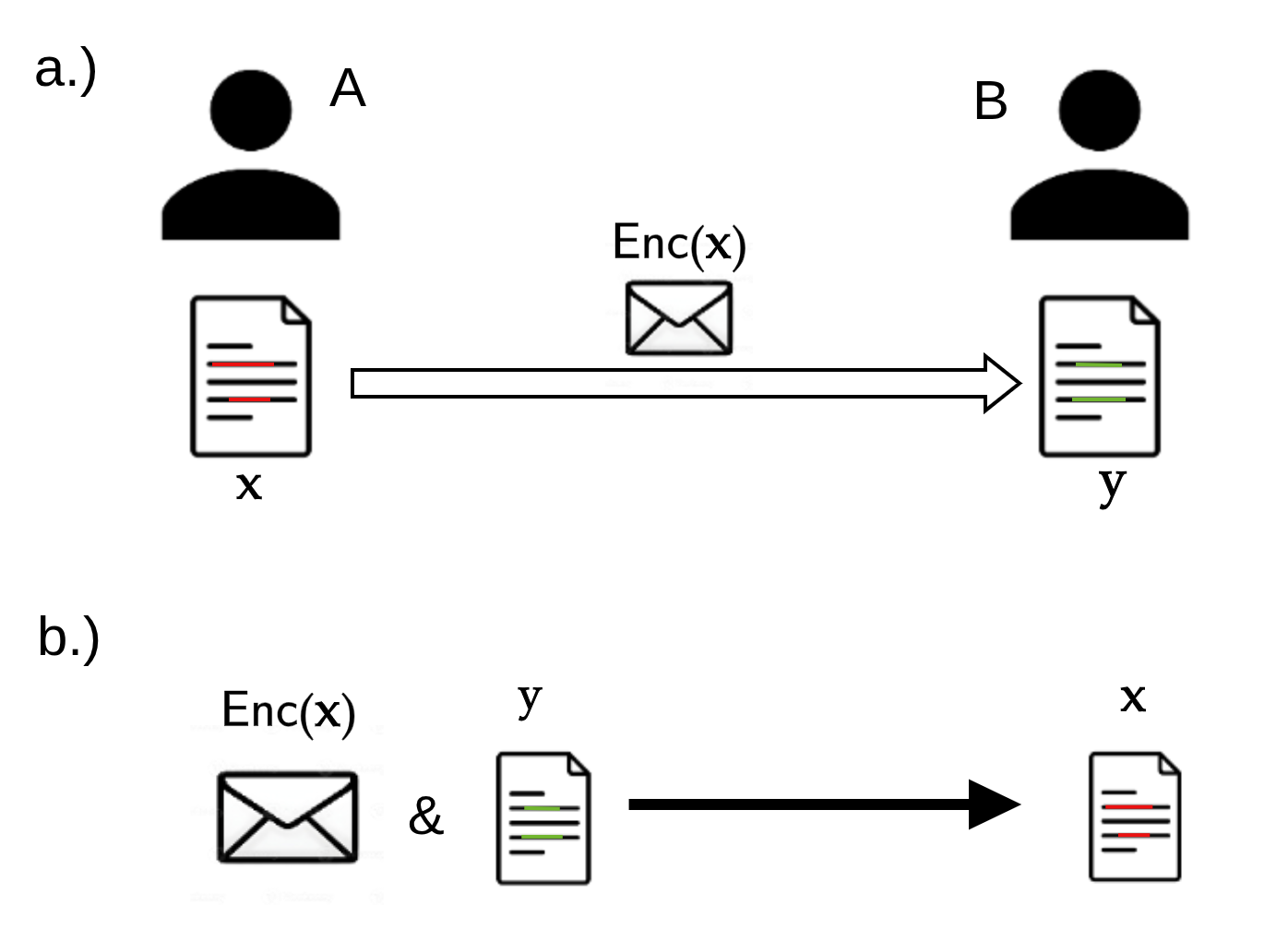}


	\caption{\small{Illustration of the Document Exchange problem.
    a.) two parties A and B have slightly different strings $\mathbf{x}$ and $\mathbf{y}$, respectively. Party A communicates $\encode(\bfx)$ to party B. b.) From $\encode(\bfx)$ and $\mathbf{y}$, party B can retrieve the string $\mathbf{x}$.}
}
	\label{fig:dx}
\end{figure}

The model typically considered in the document exchange literature assumes that $\mathbf{x}$ and $\mathbf{y}$ differ by a number of insertions, deletions, or substitutions. Such operations are called \insdel edits. The minimum number of \insdel edits needed to change $\mathbf{x}$ into $\mathbf{y}$ is bounded from above by a parameter $\tau$. 

Early works on document exchange under this model can be found in \cite{Dx_early_1, Dx_early_2, Dx_early_3, Dx_early_4, Dx_early_5, Ramachandran_docex}. 
For constant $\tau$, the best known  
document exchange scheme for at most $\tau$ \insdel edits has redundancy  $4\tau\log\bl + o(\log\bl)$ bits~\cite{Sima_optimal_t_del_code}, where $n$ denotes the length of $\mathbf{x}$.  It is easy to see that solving the document exchange problem under any edit model is equivalent to constructing systematic binary codes correcting $\tau$ edit errors. 
A systematic binary code capable of correcting $\tau$ \insdel edit errors with redundancy $4\tau\log\bl + o(\log\bl)$ bits is constructed in \cite{Sima_optimal_t_del_code}, which resulted in the above mentioned best known document exchange scheme.  
This construction utilizes a technique called syndrome compression introduced in \cite{Sima_Syndrome_compression}. Codes correcting combinations of insertions, deletions and substitutions are studied extensively in the literature (see, for instance, \cite{Guruswami_Two_del_exist_bound, Andreas_Nikita, Rawad_Serge, Antonia_Eitan, Song_Systematic_mult_del_sub_(precoding), Ye_Two_burst_excactly_b_deletion, new_1,new_2, new_3, new_4, new_5} and the references therein) due to their applications in DNA-based data storage. 
However, many of the known code constructions are non-systematic and cannot be directly used to obtain document exchange schemes. 

The authors of~\cite{Tang_substring} introduced the substring edit model defined as follows. For a given constant $k$, a $k-$substring edit applied to a string $\mathbf{x}$ results in a substring $\mathbf{u}$ of $\mathbf{x}$ being replaced by another string $\mathbf{v}$ such that the lengths of $\mathbf{u}$ and $\mathbf{v}$ are bounded from above by $k$. Substring edits is a general model which covers burst deletions, and localized insertions and deletions within a window. An analysis of real and simulated data~\cite{Tang_substring} has shown that substring edits is a good error model for document exchange.

In this work, we focus on constructing schemes for the document exchange problem with $t$ many $k$-substring edits, where $t$ is a constant. The best known document exchange scheme (in terms of redundancy) for the considered model can be constructed using the systematic codes of the recent work \cite{Li_Distributed_graph_coloring}, correcting $t$ many $k$-substring edits, $t\geq 1$. Those codes are constructed using a technique called distributed graph coloring and have redundancy $4\numedits\log\bl + O(\log\log\bl)$ bits. However, their encoding and decoding complexities are $O(n^{8t})$ and $O(n^{10t})$, respectively. Our goal is to construct schemes with similar redundancy and lower complexities.



In addition to the worst-case document exchange explained above, we consider the average-case document exchange. In the average case, multiple encoding functions with different redundancy can be used to encode $\mathbf{x}$. We consider the case where $\bfx$ is drawn uniformly at random from the set of all binary strings of length $n$. The goal is to reduce the average redundancy of the used encodings. Under a \emph{single} $k$-substring edit, the authors of \cite{Li_Optimal_one_substring_error} construct a scheme with average redundancy $\log\bl + O(\log\log\bl)$ bits. To that end, the authors introduce non-systematic codes correcting a \emph{single} $k$-substring edit error with the smallest known redundancy of $\log\bl + O(\log\log\bl)$~\cite{Li_Optimal_one_substring_error}. However, document exchange schemes for multiple $k$-substring edits in the average case remain unexplored. 


Our contributions can be summarized as follows: \begin{enumerate}[label={\emph{(\Roman*)}}]
    \item Using the syndrome compression technique, we construct a worst-case document exchange scheme with redundancy $4\numedits\log\bl + o(\log\bl)$ bits with encoding and decoding complexities of $O(n^{2t+1})$ and  $O(n^{t+1})$, respectively. 
    
    \item For the average-case document exchange problem,  
    we construct a scheme with average redundancy of $(4\numedits-1)\log\bl + o(\log\bl)$ bits. The main idea is to divide the space of input strings into two parts, which allows for different redundancies. We subsequently demonstrate that the majority of the strings come from the part that requires a small redundancy.
\end{enumerate}

The rest of the paper is organized as follows. 
In \cref{sec:background}, we present notations and introduce the problem setting. Section~\ref{sec: Dx} presents our document exchange scheme, which is based on syndrome compression. In Section~\ref{sec:Avg_Dx}, we present the average-case document exchange scheme. 

\section{Preliminaries} \label{sec:background}

We start by introducing some necessary definitions and notations. We use bold font to denote binary strings, e.g., $\mathbf{x}$. For any binary string $\mathbf{x}$, we represent its length as $\length{\mathbf{x}}$. The entries of a string $\bfx$ of length $n$ are denoted by $x_1,\cdots,x_n$, i.e., we write $\bfx=(x_1,\cdots,x_n)$. A substring $\bfu$ of $\bfx$ of length $n_\bfu \leq n$ is a string of $n_\bfu$ consecutive entries of $\bfx$, e.g., $\bfu = (x_i, \cdots, x_{n_\bfu +i -1})$, $i\in [1,n-n_\bfu+1]$. Here, the index $i$ is called the location of the substring $\bfu$ within $\bfx$. We let $(\mathbf{x}, \mathbf{y}, \cdots, \mathbf{z})$ be the concatenation of the strings $\mathbf{x}, \mathbf{y}, \cdots, \mathbf{z}$. For any $k \in \mathbb{N}$, let strings $0^k = (0,\cdots,0)$ and $1^k = (1,\cdots,1)$, where $\length{0^k} = \length{1^k} = k$. Then we denote by $0^k1^k$ the string $(0^k,1^k)$. Sets are denoted by calligraphic letters, e.g., $\cS$, and their cardinality is denoted by $|\cdot|$. For any $a, b \in \mathbb{Z}, a < b$, the set $\{a, a+1, \cdots, b\}$ is denoted as $[a,b]$. We denote by $\{0,1\}^{\leq n}$ the set of all binary strings of length less than or equal to $n$ and by $\{0,1\}^*$ the set of all binary strings. All logarithms are base $2$.

\subsection{Edit Models}
We now introduce the various edit models required to describe our problem setting and results.

For any string $\mathbf{x} \in \{0,1\}^\bl$, a valid arbitrary location $i\in[1, \bl+1]$ and an arbitrary symbol $\omega \in \{0,1\}$, a single \emph{insertion-deletion-substitution (\insdel) edit} is the operation where either $w$ is inserted in the $i$-th location, $x_i$ is substituted with $\omega$, or $x_i$ is deleted from the string. Such an operation would produce as output either $(x_1,\cdots,x_{i-1},\omega, x_{i}, \cdots, x_\bl)$, $(x_1,\cdots,x_{i-1},\omega, x_{i}, \cdots, x_\bl)$, or $(x_1,\cdots,x_{i-1},y, x_{i}, \cdots, x_\bl)$, respectively. Multiple \insdel edits imply that $\mathbf{x}$ undergoes several subsequent single \insdel edits.

A special case of multiple \insdel edits is substring edits. For any $\mathbf{x} \in \{0,1\}^\bl$, let $\bfu$ be a substring of $\bfx$ with location $i\in [1,n-\length{\bfu}+1]$. A \emph{$\editsize$-substring edit} is the operation where $\mathbf{u}$ is deleted and a string $\mathbf{v}$ is inserted in the same location $i$, and $\length{\mathbf{u}}, \length{\mathbf{v}} \leq \editsize$. Hence, the resulting string would be $((x_1,\cdots,x_{i-1}), \bfv, (x_{i+\length{\bfu}}, \cdots, x_\bl))$.
\begin{example}
    Consider $\mathbf{x} = (1110{\color{red}0010}100110)$ and $\mathbf{y} = (1110{\color{blue}10011}100110)$. From $\mathbf{x}$ one can get $\mathbf{y}$ through a $5$-substring edit with $\mathbf{u} = ({\color{red}0010})$ and $\mathbf{v} = ({\color{blue}10011})$. 
\end{example}

If a $\editsize$-substring edit operation has been applied on $\mathbf{x}$ subsequently over $\numedits$ rounds, the string $\mathbf{x}$ is said to have undergone $\numedits$ $\editsize$-substring edits. We will primarily focus on the document exchange under this edit model, which is also considered in \cite{Li_Distributed_graph_coloring}. Note that $(\numedits-1)$ $\editsize$-substring edit is a special case of $\numedits$ $\editsize$-substring edit, where the deleted and inserted strings are \emph{identical} in at least one round.




We will now introduce some notions connected to application of the edit models on input strings. For any string $\mathbf{x}$, the \emph{edit ball} of $\mathbf{x}$ under $\numedits$ $\editsize$-substring edits $\editball{\numedits}{\mathrm{str}(\editsize)}{\bfx}$  is the set of strings that can be produced from $\mathbf{x}$ by applying at most $\numedits$ $\editsize$-substring edits. For two equal length strings $\mathbf{x}$ and $\mathbf{y}$, we say they are \emph{confusable under $\numedits$ $\editsize$-substring edits} if both strings can produce the same output string under $\numedits$ $\editsize$-substring edits, i.e., $\editball{\numedits}{\mathrm{str}(\editsize)}{\bfx} \cap \editball{\numedits}{\mathrm{str}(\editsize)}{\bfy} \neq \emptyset$. The \emph{confusion ball under $\numedits$ $\editsize$-substring edits} of $\mathbf{x} \in \{0,1\}^\bl$, denoted by $\confball{\numedits}{\mathrm{str}(\editsize)}{\mathbf{x}}$, is the set of strings that are confusable with $\mathbf{x}$ under $\numedits$ $\editsize$-substring edits, namely $\confball{\numedits}{\mathrm{str}(\editsize)}{\mathbf{x}} = \{\bfy \in \{0,1\}^\bl : \editball{\numedits}{\mathrm{str}(\editsize)}{\bfx} \cap \editball{\numedits}{\mathrm{str}(\editsize)}{\bfy} \neq \emptyset\}$.
Further, we define the \emph{edit ball} of $\mathbf{x}$ under $t$ \insdel edits as the set of strings that can be produced from $\mathbf{x}$ by applying $t$ \insdel edits to $\bfx$. This is  denoted as $\editball{t}{\mathrm{\insdel}}{x}$.
For a general statement on the confusion ball which holds true for any choice of $\numedits$ and $\editsize$, we drop the subscripts and superscripts and simply denote the confusion ball as $\confball{}{}{x}$.

\subsection{Problem Statement}
In this paper, for our document exchange schemes, we focus on the multiple-substring edit model, in which strings $\mathbf{x}$ and $\mathbf{y}$ differ by $\numedits$ $\editsize$-substring edits. Such a scenario is formally described below.

\begin{definition}[Document Exchange]
     Consider a string $\mathbf{x} \in \{0,1\}^\bl$ and a string $\mathbf{y}\in\editball{\numedits}{\mathrm{str}(\editsize)}{\mathbf{x}}$. An \emph{encoding function} with redundancy $s$, $\encode:\{0,1\}^\bl \to \{0,1\}^s$, produces a valid \emph{document exchange scheme} if $\bfx$ can be deterministically determined by solely observing $\mathbf{y}$ and $\encode(\mathbf{x})$.
\end{definition}

The complexities of computing the encoding function and determining $\mathbf{x}$ from $\mathbf{y}$ and $\encode(\mathbf{x})$ are referred to as the \emph{encoding} and \emph{decoding} complexities, respectively.

\begin{definition}[Average-case Document Exchange]
     Let the sets $\cP_1,\cdots,\cP_r$ form a partition of $\{0,1\}^n$ and let $\encode_j:\cP_j \to \{0,1\}^{s_j}$, $j\in [1,r]$ be $r$ different encoding functions each with redundancy $s_j$. The $r$ encoding functions produce a valid \emph{average-case document exchange scheme} if any $\bfx \in \{0,1\}^\bl$ can be deterministically determined by solely observing a string $\bfy\in\editball{\numedits}{\mathrm{str}(\editsize)}{\mathbf{x}}$, $r$ and $\encode_r(\bfx)$. Assuming all strings $\bfx$ have the same probability of being chosen, the average redundancy of the scheme is defined as $\bar{s} = \frac{\sum_{j=1}^r |\cP_j|s_j}{n}.$
\end{definition}

By using different encoding functions for different input strings, we may be able to reduce the expected redundancy overall all strings. 
The paper \cite{Li_Optimal_one_substring_error} considers document exchange under one substring edit for cases when $p_0 \neq p_1$, where $p_0$ and $p_1$ are the probabilities corresponding to each bit of $\bfx$ being $0$ or $1$. However, we do not consider such scenarios.



\section{Document exchange under multiple substring edits}
\label{sec: Dx}
In this section, we will present the document exchange scheme under $\numedits$ $\editsize$-substring edits. Using the syndrome compression technique~\cite{Sima_optimal_t_del_code}, we construct an encoding function for a document exchange scheme with redundancy $4\numedits\log\bl + o(\log\bl)$ bits. The encoding and decoding complexities of our scheme are $O(\bl^{2\numedits+1})$ and $O(\bl^{\numedits+1})$, respectively.

Syndrome compression is a general technique for constructing low-redundancy error correcting codes 
as described formally below.

\begin{lemma}[{\protect\cite[Section~2]{Sima_Syndrome_compression}}]
\label{lem:labeling function}
    Consider an arbitrary error model and for $\bfc\in \{0,1\}^n$, let $\confball{}{}{\mathbf{c}}$ represent the confusion ball of $\mathbf{c}$ under the edit model in consideration.
    
    Define a function $f:\{0,1\}^\bl \to \{0,1\}^R$ satisfying
    \begin{enumerate}
        \item $f(\mathbf{c}) \neq f(\mathbf{c}')$, for any $\mathbf{c}, \mathbf{c}' \in \{0,1\}^\bl$ and  $\mathbf{c}' \in \confball{}{}{\mathbf{c}}$, and,
        \item $R \leq O((\log\log\bl)\log\bl)$.
    \end{enumerate}
    We refer to such a function $f$ as a \emph{labeling function}.
    Then, for any fixed $\mathbf{x} \in \{0,1\}^\bl$, there exists an integer $$a_\mathbf{x}\leq2^{\log{|\confball{}{}{\mathbf{x}}|}+o(\log\bl)}$$ such that $$ f(\mathbf{x}) \not\equiv f(\mathbf{y}) \pmod{a_\mathbf{x}}, \quad \forall \ \mathbf{y} \in \confball{}{}{\mathbf{x}}.$$
    The time complexity for finding $a_\bfx$ corresponding to a given $\mathbf{x}$ is at most $O(|\confball{}{}{\mathbf{x}}|^2)$.
\end{lemma}

The syndrome compression technique requires that the labeling function maps to different outputs for all mutually confusable strings. Our edit model here is \tksub edits, which are a special case of multiple \insdel edits. Hence, we choose the labeling function from~\cite{Sima_optimal_t_del_code}, capable of correcting multiple \insdel edits.
\begin{lemma}[{\protect\cite[Section~5]{Sima_optimal_t_del_code}}]
\label{lem:edit code}
    For any $\tau$, constant w.r.t.\ $\bl$, there exists a labeling function $f:\{0,1\}^\bl \to \{0,1\}^R$, where $$R~\leq [(\tau^2+1)(2\tau^2+1)+2\tau^2(\tau-1)]\log\bl + o(\log\bl),$$ such that for any $\mathbf{c}, \mathbf{c}' \in \{0,1\}^\bl$ and  $\mathbf{c}' \in \confball{\tau}{\mathrm{\insdel}}{\mathbf{c}}$, it holds that $f(\mathbf{c})\neq f(\mathbf{c}')$. 
\end{lemma}

Further, the complexity of computing the function $f$ is at most $O(\bl)$ \cite{Sima_optimal_t_del_code}.
As $\numedits$ $\editsize$-substring edits are equivalent to at most $2\numedits\editsize$ \insdel edit errors, by appropriately substituting $\tau = 2\numedits\editsize$ in \cref{lem:edit code}, we get the desired labeling function for our \tksub model. 

For any $\bfx \in \{0,1\}^\bl$ and a labeling function $f$, document exchange schemes using the syndrome compression technique typically have encoding functions of the form $\encode(\bfx) = (f(\bfx) \mod {a_\bfx}, ~a_\bfx)$. The redundancy of such an encoding function is given by the maximum value of $a_\bfx = O(|\cC(\bfx)|)$. Thus, finding the redundancy of is related to finding the cardinality of the largest confusion ball, i.e., $\max_{\bfx \in \{0,1\}^n} |\cC(\bfx)|$. For our model, we formally quantify this value in \cref{lem:confusion ball general case}.
\begin{lemma}
\label{lem:confusion ball general case}
    For any string $\mathbf{x} \in \{0,1\}^\bl$ and for any $\numedits,\editsize$ as fixed constants, the following upper bound holds
    \begin{align*}
        |\confball{\numedits}{\mathrm{str}(\editsize)}{\mathbf{x}}| = O(\bl^{2\numedits}).
    \end{align*}  
\end{lemma}
\begin{proof}
    By definition, for any $ \mathbf{y} \in \confball{\numedits}{\mathrm{str}(\editsize)}{\mathbf{x}}$, it holds that $$\editball{\numedits}{\mathrm{str}(\editsize)}{\mathbf{y}} \cap \editball{\numedits}{\mathrm{str}(\editsize)}{\mathbf{x}}\neq \emptyset.$$ Let $\mathbf{z} \in \editball{\numedits}{\mathrm{str}(\editsize)}{\mathbf{y}} \cap \editball{\numedits}{\mathrm{str}(\editsize)}{\mathbf{x}}$. Since $\mathbf{z}$ is at most \tksub edits away from either of $\bfx$ or $\bfy$, it holds that $\bfx$ and $\bfy$ are at most $2\numedits$ substring edits away from each other. Thus, 
    \begin{align}
    \label{eq: confball_subset}
        \confball{\numedits}{\mathrm{str}(\editsize)}{\mathbf{x}} \subseteq \editball{2\numedits}{\mathrm{str}(k)}{\mathbf{x}}.
    \end{align}
    
    A $\editsize$-substring edit on $\mathbf{x}$ is equivalent to deleting a substring $\mathbf{u}$ from $\mathbf{x}$ and inserting a substring $\mathbf{v}$ in the same location, where $\length{\mathbf{u}}, \length{\mathbf{v}} \leq \editsize$. Thus, as $\length{\mathbf{x}} = n$, the number of strings that are \tksub edits away from $\mathbf{x}$, i.e., $\left|\editball{1}{\mathrm{str}(k)}{\mathbf{x}}\right|$, is at most $\bl\editsize(\sum_{i=0}^\editsize 2^i) \leq \bl\editsize3^\editsize$. Note that the length of the resulting strings is at most $\bl+\editsize$. 
    Similarly, for $2\numedits$ $\editsize$-substring edits, we have
    \begin{equation}\label{eq:ballsize}
        \left|\editball{2\numedits}{\mathrm{str}(k)}{\mathbf{x}}\right| \leq \left(\prod_{i=0}^{2\numedits}(n+i\editsize)\right)\cdot\editsize^{2\numedits} 3^{2\editsize\numedits}.
    \end{equation}
    Therefore, $\left|\editball{2\numedits}{\mathrm{str}(k)}{\mathbf{x}}\right|= O(\bl^{2t})$, and  using Equation (\ref{eq: confball_subset}), 
    \begin{equation*}
        \left|\confball{\numedits}{\mathrm{str}(\editsize)}{\mathbf{x}}\right| = O(\bl^{2t}).\qedhere
    \end{equation*}
\end{proof}

Using the above results, we now present our document exchange scheme under \tksub edits and analyze its redundancy and complexities.
\begin{theorem}[Document Exchange Scheme]
\label{thm: Dx}
    Consider $\mathbf{x} \in \{0,1\}^\bl$ and let $\numedits, \editsize$ be fixed constants and $f$ be the labeling function in Lemma~\ref{lem:edit code}. Then, for $a_\mathbf{x}$ as in Lemma~\ref{lem:labeling function} and any $\bfy\in\editball{\numedits}{\mathrm{str}(\editsize)}{\mathbf{x}}$, the following encoding function
    \begin{align*}
        \encode(\mathbf{x}) = \bar{f}(\mathbf{x}) \defeq (f(\mathbf{x})\mod {a_\mathbf{x}},~a_\mathbf{x})
    \end{align*}
    produces a valid document exchange scheme with redundancy $4\numedits\log\bl + o(\log\bl)$ bits and complexity of encoding and decoding as $O(\bl^{2\numedits+1})$ and $O(\bl^{\numedits+1})$, respectively.  
\end{theorem}

\begin{proof}

\emph{Determining $\bfx$ correctly:}  By definition, the set of strings of length $\bl$ that are $\numedits$ $\editsize$-substring edits of $\mathbf{y}$ is a subset of $ \confball{\numedits}{\mathrm{str}(\editsize)}{\mathbf{x}}$. From Lemma \ref{lem:labeling function}, for any string $\mathbf{z} \in \confball{\numedits}{\mathrm{str}(\editsize)}{\mathbf{x}}$, it holds that $f(\mathbf{x}) \not\equiv f(\mathbf{z}) \pmod{a_\mathbf{x}}$. Thus, given $\encode(\mathbf{x}) = (f(\mathbf{x})\mod{a_\mathbf{x}},~a_\mathbf{x})$, we can uniquely identify $\mathbf{x}$ by iteratively computing $f(\hat{\mathbf{x}}) \mod{a_\mathbf{x}}$ for all $\hat{\mathbf{x}} \in \editball{\numedits}{\mathrm{str}(k)}{\bfy}$. 
  
\emph{Decoding complexity:} Given $\bfy$, using an analogous bound to the one in \cref{eq:ballsize} (replacing $2\numedits$ in \cref{eq:ballsize} by $t$), it can be seen that computing the set $\editball{\numedits}{\mathrm{str}(k)}{\bfy}$ requires $O(\bl^{\numedits})$ operations, because there are at most $O(\bl^{\numedits})$ many strings in this set. Further, the complexity of calculating $f(\hat{\mathbf{x}})$ is $O(\bl)$ \cite{Sima_optimal_t_del_code, Song_Systematic_mult_del_sub_(precoding)}. Thus, the overall decoding complexity is at most $O(\bl^{\numedits+1})$. 

\emph{Encoding complexity:}  To compute $\encode(\bfx)$, from Lemma \ref{lem:labeling function},  we know that we need to find an appropriate $a_\bfx$, such that $f(\mathbf{x}) \not\equiv f(\mathbf{y}) \pmod{a_\mathbf{x}}, \quad \forall \ \mathbf{y} \in  \confball{\numedits}{\mathrm{str}(\editsize)}{\mathbf{x}}$. Therefore, we need to compute the set $ \confball{\numedits}{\mathrm{str}(\editsize)}{\mathbf{x}}$, which takes $O(\bl^{2\numedits})$ operations. For each $\bfy \in \confball{\numedits}{\mathrm{str}(\editsize)}{\mathbf{x}}$, we need to calculate $f(\bfy)$, which takes $O(\bl)$ operations. The value of $a_\bfx$ needed for the property $f(\bfx)\not\equiv f(\bfy) \pmod*{a_\bfx}$ to hold can be computed in parallel to computing $f(\bfy)$ at no extra cost, see~\cite{Sima_Syndrome_compression}. The reason is that one can start with $a_\bfx = 2$ and increase it whenever we find a $\bfy$ with $f(\bfy)\equiv f(\bfx) \pmod*{a_\bfx}$. Hence, the overall complexity of calculating $\encode(\mathbf{x})$ is $O(n^{2\numedits+1})$.
  
\emph{Redundancy:} From \cref{lem:labeling function,lem:confusion ball general case}, one needs at most $2t\log\bl + o(\log\bl)$ bits to store $a_\mathbf{x}$ and $f(\bfx) \mod a_\bfx$. Thus, the total length of $\encode(\mathbf{x})$ is $4t\log\bl + o(\log\bl)$ bits.
\end{proof}

In \cite{Li_Distributed_graph_coloring}, for constant  $\numedits$ and $\editsize$, a document exchange scheme for $\numedits$ $\editsize$-substring edits with redundancy $4\numedits\log\bl + O(\log\log\bl)$ bits is constructed using distributed graph coloring. A key advantage of the scheme in~\cite{Li_Distributed_graph_coloring} is that it does not require a prior labeling function. Therefore, it is applicable even when $\editsize$ scales with $\bl$. However, the complexities of encoding and decoding are $O(\bl^{8\numedits})$ and $O(\bl^{10\numedits})$, respectively, which are much higher than those demonstrated by our scheme. Furthermore, both schemes have the same dominant term in the redundancy.

\section{Average-case document exchange under multiple substring edits }
\label{sec:Avg_Dx}
We present our average-case document exchange scheme with an average redundancy of $(4\numedits-1)\log\bl + o(\log\bl)$ bits. 

The main idea behind average-case document exchange is that the redundancy of the $\encode$ function of a worst-case document exchange scheme is determined by the largest confusion ball. However, only a few strings have a large cardinality of the confusion ball. The remaining majority have a small cardinality of the confusion ball and will require encodings with smaller redundancy. 

We consider strings $\bfx$ that are \emph{pattern-dense}, i.e.,  each substring of $\bfx$ of length $\delta$ must contain a fixed string $\bfp$, which we refer to as \emph{pattern}. 
\begin{definition}[Pattern-dense strings]
        For some fixed pattern $\mathbf{p}\in\{0,1\}^{\leq \delta}$, any string $\mathbf{x} \in \{0,1\}^\bl$ is \emph{$(\mathbf{p}, \delta)$-dense} if every substring of $\mathbf{x}$ of length $\delta$ contains the pattern $\bfp$ as a substring.
\end{definition}

The goal is to demonstrate that most of the strings in $\{0,1\}^\bl$ are pattern-dense. We then show that such strings require an encoding with a small redundancy. 


    
   

We now demonstrate that for any fixed pattern $\bfp$ and large enough $\delta$, the number of strings in $\{0,1\}^\bl$ which are not $(\mathbf{p}, \delta)$-dense is vanishingly small, as $\bl \to \infty$.
\begin{lemma}
\label{lem:pattern enrichment}
Let $\bfx$ be sampled uniformly at random from $\{0,1\}^\bl$ and let $\delta = \alpha \log \bl$ for some fixed constant $\alpha$. Then, for any pattern $\bfp$ with length constant with respect to $\bl$,
\begin{align*}
    \Pr(\mathbf{x} \text{ is not } (\mathbf{p}, \delta)\text{-dense}) = O\left(\frac{1}{n^{\frac{\lfloor\alpha/\length{\mathbf{p}}\rfloor}{2^{\length{\mathbf{p}}}}-1}}\right).
\end{align*}
\end{lemma}

\begin{proof}
 If $\mathbf{x}$ is not $(\mathbf{p}, \delta)$-dense, it means that there exists at least one index $i$, such that $p\neq(x_j,x_{j+1}, \cdots, x_{j+\length{\mathbf{p}}-1}), \forall j \in [i,i+\delta-1]$. Thus,
 
    \begin{align*}
        \Pr(\mathbf{x} \text{ is not } (\mathbf{p}, \delta)\text{-dense}) &=  \Pr\left(\bigcup_{i=1}^{\bl-\delta+1} \{\bfp\neq(x_j,x_{j+1}, \cdots, x_{j+\length{\mathbf{p}}-1}), \forall j \in [i,i+\delta-1]\}\right)\\
        &\stackrel{(a)}{\leq}  \sum_{i=1}^{\bl-\delta+1} \Pr\left(\mathbf{p}\neq(x_j,x_{j+1}, \cdots, x_{j+\length{\mathbf{p}}-1}), \forall j \in [i,i+\delta-1]\right)\\
        &= (\bl-\delta+1)\cdot\Pr\left(\mathbf{p}\neq(x_j,x_{j+1}, \cdots, x_{j+\length{\mathbf{p}}-1}), \forall j \in [1,\delta]\right)\\
        &=(\bl-\delta+1)\cdot\bigcap_{j=1}^{\delta}\Pr(\mathbf{p}\neq(x_j,x_{j+1}, \cdots, x_{j+\length{\mathbf{p}}-1}))\\
        &\stackrel{(b)}{\leq} (\bl-\delta+1)\cdot\prod_{j=1}^{\lfloor\delta/\length{\mathbf{p}}\rfloor}\Pr(p\neq(x_{(j-1)\cdot\length{\mathbf{p}}+1}, \cdots, x_{j\cdot\length{\mathbf{p}}}))\\
        &=  (\bl-\delta+1)\cdot \prod_{j=1}^{\lfloor\delta/\length{\mathbf{p}}\rfloor}(1-\Pr(p\neq(x_{(j-1)\cdot\length{\mathbf{p}}+1}, \cdots, x_{j\cdot\length{\mathbf{p}}})))\\
        &= (\bl-\delta+1)\cdot \left(1-\frac{1}{2^{\length{\mathbf{p}}}}\right)^{\lfloor\delta/\length{\mathbf{p}}\rfloor}\\
        &\stackrel{(c)}{\leq} (\bl-\delta+1)\cdot \exp\left(\frac{-\lfloor\delta/\length{\mathbf{p}}\rfloor}{2^{\length{\mathbf{p}}}}\right)\\
        &= O\left(\frac{1}{n^{\frac{\lfloor\alpha/\length{\mathbf{p}}\rfloor}{2^{\length{\mathbf{p}}}}-1}}\right).
    \end{align*}
Here, $(a)$ follows from the union bound; $(b)$ holds because the probability of the intersection of all events in a set is upper bounded by the probability of the intersection of any subset of events, which for a subset containing mutually independent events equals the product of their probabilities. Finally, $(c)$ holds as $\forall x \in \mathbb{N}, (1-x)\leq e^{-x}$.
\end{proof}
In \cite{Li_Optimal_one_substring_error}, for some choice of $\bfp$ and $\delta$, Li \emph{et al.} demonstrate a scheme capable of correcting one $\editsize$-substring edit for any $(\mathbf{p}, \delta)$-dense string. A description is provided below.
\begin{lemma}[{\protect\cite[Theorem~19]{Li_Optimal_one_substring_error}}]
\label{lem: substring edit code}
Consider any constant $\editsize$, $\delta = \editsize2^{2\editsize+3}\log\bl$ and $\bfp=0^\editsize1^\editsize$. Then, there exist functions $h$ and $\psi$, such that any $(\mathbf{p}, \delta)$-dense string  $\mathbf{x} \in\{0,1\}^\bl$ can be uniquely decoded from
   $ (\bfy, \mathbf{h(\mathbf{x})}, \mathbf{\psi(\mathbf{x})}),$
for any  $\mathbf{y} \in \editball{1}{\mathrm{str}(\editsize)}{\mathbf{x}}$ and where $\length{\mathbf{h(\mathbf{x})}}\leq \log\bl$ and $\length{\psi(\mathbf{x}}) = o(\log\bl)$. The overall complexity of computing $\mathbf{h(\mathbf{x})}$, $\mathbf{\psi(\mathbf{x})}$ and the decoding process is at most $O(\bl \ \mathrm{ poly }\log \bl)$.
\end{lemma}
For document exchange schemes using syndrome compression or distributed graph coloring, the redundancy depends on the size of the confusion ball. In our scheme, we partition $\{0,1\}^n$ into two parts $\cP_1$ and $\cP_2$ such that all $\bfx \in \cP_1$ are $(\bfp,\delta)$-dense strings and all $\bfx\in \cP_2$ are not. This partitioning allows us to encode with small average redundancy, provided we also convey the partition index of $\bfx$ with $\encode(\bfx)$. As such, we can define a \emph{restricted confusion ball} as follows.


\begin{definition}
    Consider a fixed pattern $\mathbf{p}$ and a parameter $\delta$. Then, for every $(\mathbf{p}, \delta)$-dense string $\mathbf{x} \in \{0,1\}^\bl$, the \emph{restricted confusion ball} is defined as,
    \begin{align*}
       \resconfball{\numedits}{\mathrm{str}(\editsize)}{\bfx} \defeq \{\bfy \in \confball{\numedits}{\mathrm{str}(\editsize)}{\bfx}: \bfy \text{ is } (\mathbf{p}, \delta)\text{-dense}\}.
    \end{align*}
\end{definition}

The idea of partitioning the confusion ball based on prior information about the string $\bfx$ is connected to the so-called \emph{pre-coding technique} described in \cite{Song_Systematic_mult_del_sub_(precoding)}. There, Song \emph{et al.}~\cite{Song_Systematic_mult_del_sub_(precoding)} allow the encoding of strings that are codewords of some selected code. This prior information is used to restrict the confusion ball and obtain a smaller redundancy. However, a coding scheme based on pre-coding does not directly translate to an average-case document exchange scheme, because codewords are sparsely distributed in the ambient space. 

We take inspiration from this technique and consider a family of codes to encode the strings in $\cP_1$. In particular, we consider the family of one $\editsize$-substring edit correcting codes proposed in \cite{Li_Optimal_one_substring_error} and using $\bfp = 0^\editsize1^\editsize$, we show that all strings in $\cP_1$ have a reduced size of the restricted confusion ball.

\begin{lemma}
    \label{lem: reduction of confusable ball size}
    For any constant $\editsize$, $\delta = \editsize2^{2\editsize+3}\log\bl$ and $\bfp=0^\editsize1^\editsize$, the size of the restricted confusion ball of any $(\mathbf{p}, \delta)$-dense string $\mathbf{x} \in \{1,0\}^\bl$ can be reduced to
    \begin{align*}
        |\resconfball{\numedits}{\mathrm{str}(\editsize)}{\mathbf{x}}| = O(\bl^{2\numedits-1}).
    \end{align*}  
    by leveraging the fact that $\bfx$ is $(\bfp,\delta)$-dense and using $\mathbf{h(\mathbf{x})}$ and $\psi(x)$ as described in Lemma \ref{lem: substring edit code}.
\end{lemma}

\begin{proof}
  The proof is similar to that of Lemma \ref{lem:confusion ball general case}. Since $\mathbf{x}$ is known to be $(\mathbf{p}, \delta)$-dense, in such a case, we need to only consider the restricted confusion ball $\resconfball{\numedits}{\mathrm{str}(\editsize)}{\mathbf{x}}$. Since $\resconfball{\numedits}{\mathrm{str}(\editsize)}{\mathbf{x}} \subset \confball{\numedits}{\mathrm{str}(\editsize)}{\mathbf{x}}$, and, from Lemma \ref{lem:confusion ball general case}, $\confball{\numedits}{\mathrm{str}(\editsize)}{\mathbf{x}} \subseteq \editball{2\numedits}{str}{\mathbf{x}}$, it holds that, $$\resconfball{\numedits}{\mathrm{str}(\editsize)}{\mathbf{x}} \subset  \editball{2\numedits}{\mathrm{str}(\editsize)}{\mathbf{x}}\cap\{\mathbf{y}\in\{0,1\}^\bl : \mathbf{y} \text{ is } (\mathbf{p}, \delta)\text{-dense}\}.$$
  
  Now, from Lemma \ref{lem: substring edit code}, given $\mathbf{h(\mathbf{x})}$ and  $\mathbf{\psi(\mathbf{x})}$, one can uniquely decode $\mathbf{x}$ from any string that is one $\editsize$-substring edit away from $\mathbf{x}$. Thus, the size of the restricted confusion ball of $\mathbf{x}$ is  bounded from above by the number of strings in $\editball{2\numedits-1}{\mathrm{str}(\editsize)}{\mathbf{x}}$ that uniquely decode to an $\bl$-length $(\mathbf{p}, \delta)$-dense string. Hence, using arguments similar to Lemma \ref{lem:confusion ball general case}, we have, \begin{equation*}
      |\resconfball{\numedits}{\mathrm{str}(\editsize)}{\mathbf{x}}| \leq |\editball{2\numedits-1}{\mathrm{str}(\editsize)}{\mathbf{x}}| = O(\bl^{2\numedits-1}).\qedhere
  \end{equation*}
\end{proof}

The above result on restricted confusion ball sizes allows us to reduce the redundancy for most strings. With this at hand, we present our average-case document exchange scheme.
\begin{theorem}[Average-case Document Exchange]
\label{thm: ADx}
    Let any string $\mathbf{x} \in \{0,1\}^\bl$, $\numedits, \editsize$ be fixed constants, and  $f$ be the labeling function in Lemma \ref{lem:edit code}. Then, for $a_\mathbf{x}$ as in Lemma \ref{lem:edit code} and any $\bfy\in\editball{\numedits}{\mathrm{str}(\editsize)}{\mathbf{x}}$, the following encoding produces an average-case document exchange scheme
     \begin{align*}
    \encode(\mathbf{x}) =
    \begin{cases}
        (\mathbf{h(\mathbf{x})}, \mathbf{\psi(\mathbf{x})}, \bar{f}(\mathbf{x})), \text{ if } \bfx \text{ is } (\mathbf{p}, \delta)\text{-dense}\\
        \bar{f}(\mathbf{x}), \text{ otherwise,}
    \end{cases}
    \end{align*} 
    with expected redundancy $(4\numedits-1)\log\bl + o(\log\bl)$ bits, where $\bfp = 0^\editsize1^\editsize$, $\delta = \editsize2^{2\editsize+3}\log\bl$, and $\mathbf{h(\mathbf{x})}, \mathbf{\psi(\mathbf{x})}$ are as described in Lemma \ref{lem: substring edit code}.
\end{theorem}

\begin{proof}
Since the structure of $\encode(\mathbf{x})$ is different depending on whether $\mathbf{x}$ is $(\mathbf{p}, \delta)$-dense or not, we can clearly distinguish between the two cases. 

If $\mathbf{x}$ is $(\mathbf{p}, \delta)$-dense, we only need to consider the restricted confusion ball $\resconfball{\numedits}{\mathrm{str}(\editsize)}{\mathbf{x}}$. Given $\mathbf{h(\mathbf{x})}$ and $\mathbf{\psi(\mathbf{x})}$, from Lemma \ref{lem: reduction of confusable ball size}, we have $$|\resconfball{\numedits}{\mathrm{str}(\editsize)}{\mathbf{x}}| = O(\bl^{2\numedits-1}).$$  We use steps similar to Theorem {\ref{thm: Dx}}, and note that since the size of $\encode(\mathbf{x})$ now depends on the size of the restricted confusion ball, the redundancy is at most $(4\numedits-2)\log\bl + o(\log\bl)$ bits. Since the length of $h(\bfx)$ is $\log n$ and the length of $\mathbf{\psi(\mathbf{x})}$ is $o(\log n)$, and they both also need to be communicated, the total redundancy of the encoding becomes at most $(4\numedits-1)\log\bl + o(\log\bl)$ bits.

Alternatively, if $\bfx$ is not $(\mathbf{p}, \delta)$-dense, we employ the procedure described in Theorem \ref{thm: Dx}. Now, from Lemma \ref{lem:pattern enrichment}, substituting $\alpha = k2^{2k+3}$ and $\length{\mathbf{p}} = 2k$, we have, $$\Pr(\mathbf{x} \text{ is not } (\mathbf{p}, \delta)\text{-dense}) = O\left(\frac{1}{n^{3}}\right)$$. Thus, on average, the length of $\encode(\mathbf{x})$ will be at most $(4\numedits-1)\log\bl + o(\log\bl)$ bits.
\end{proof}

    As the redundancy of both syndrome compression and distributed graph coloring methods depends solely on the sizes of the confusion balls, we could construct an average-case document exchange scheme with a similar expected redundancy of $(4\numedits-1)\log\bl+O(\log\log\bl)$ using the distributed graph coloring method instead of syndrome compression. However, the encoding and decoding complexities would be much higher for such a scheme.


\section{Conclusion}
In this work, we constructed a document exchange scheme for the \tksub edit model with redundancy $4\numedits \log\bl + o(\log \bl)$ and encoding and decoding complexities of $O(n^{2\numedits+1})$ and $O(n^{\numedits+1})$, respectively. Further, we addressed the previously unexamined case of \tksub edits in the average-case document exchange and proposed a scheme with average redundancy $(4\numedits -1)\log\bl + o(\log \bl)$ bits. 

 \clearpage 



\bibliographystyle{IEEEtran}
\bibliography{main.bib}

\end{document}